\documentclass[pdflatex,sn-mathphys-num]{sn-jnl}

\usepackage{graphicx}%
\usepackage{multirow}%
\usepackage{amsmath,amssymb,amsfonts}%
\usepackage{amsthm}%
\usepackage{mathrsfs}%
\usepackage[title]{appendix}%
\usepackage{xcolor}%
\usepackage{textcomp}%
\usepackage{manyfoot}%
\usepackage{booktabs}%
\usepackage{algorithm}%
\usepackage{algorithmicx}%
\usepackage{algpseudocode}%
\usepackage{listings}%
\usepackage{comment}

\theoremstyle{thmstyleone}%
\newtheorem{theorem}{Theorem}
\theoremstyle{thmstyletwo}%

\theoremstyle{thmstylethree}%

\begin{document}

\title[Article Title]{A Differentially Private Kaplan-Meier Estimator for Privacy-Preserving Survival Analysis}


\author[1]{\fnm{Narasimha Raghavan} \sur{Veeraragavan}}\email{Narasimha.Raghavan@kreftregisteret.no}

\author[2]{\fnm{Sai Praneeth} \sur{Karimireddy}}\email{karimire@usc.edu}

\author[1,3]{\fnm{Jan Franz} \sur{Nygård}}\email{Jan.Nygard@kreftregisteret.no}

\affil[1]{\orgdiv{Cancer Registry of Norway}, \orgname{Norwegian Institute of Public Health}, \orgaddress{ \city{Oslo},  \country{Norway}}}

\affil[2]{\orgdiv{Department of Computer Science}, \orgname{University of Southern California}, \orgaddress{\city{Los Angeles},  \country{USA}}}

\affil[3]{\orgdiv{Department of Physics and Technology}, \orgname{The Arctic University of Norway}, \orgaddress{\city{Tromsø},  \country{Norway}}}


\abstract{
This paper presents a differentially private approach to Kaplan-Meier estimation that achieves accurate survival probability estimates while safeguarding individual privacy. The Kaplan-Meier estimator is widely used in survival analysis to estimate survival functions over time, yet applying it to sensitive datasets, such as clinical records, risks revealing private information. To address this, we introduce a novel algorithm that applies time-indexed Laplace noise, dynamic clipping, and smoothing to produce a privacy-preserving survival curve while maintaining the cumulative structure of the Kaplan-Meier estimator. By scaling noise over time, the algorithm accounts for decreasing sensitivity as fewer individuals remain at risk, while dynamic clipping and smoothing prevent extreme values and reduce fluctuations, preserving the natural shape of the survival curve.

Our results, evaluated on the NCCTG lung cancer dataset, show that the proposed method effectively lowers root mean squared error (RMSE) and enhances accuracy across privacy budgets ($\epsilon$). At $\epsilon = 10$, the algorithm achieves an RMSE as low as 0.04, closely approximating non-private estimates. Additionally, membership inference attacks reveal that higher $\epsilon$ values (e.g., $\epsilon \geq 6$) significantly reduce influential points, particularly at higher thresholds, lowering susceptibility to inference attacks. These findings confirm that our approach balances privacy and utility, advancing privacy-preserving survival analysis.
}

\maketitle

\section{Introduction}

The Kaplan-Meier (KM) estimator~\cite{kaplan1958nonparametric}~\cite{goel2010understanding} is widely used in survival analysis to estimate the probability of survival over time in the presence of censored data. Given a sequence of time points where events (e.g., death) are observed, the KM estimator provides a step-wise estimate of survival probabilities. However, when dealing with sensitive data, such as medical or personal records, applying the KM estimator without proper privacy safeguards can reveal individual information.

Differential Privacy (DP)~\cite{dwork2006calibrating}~\cite{dwork2014algorithmic} is a formal framework for protecting individual privacy in statistical outputs. In this report, we introduce a method to compute the KM estimator in a differentially private manner, ensuring privacy guarantees for individuals in the dataset. We describe the application of DP mechanisms, particularly focusing on the addition of Laplace noise, clipping, and moment accounting to manage cumulative privacy loss over time.

\section{Background}

In this section, we provide an overview of Kaplan-Meier estimation and differential privacy, two core concepts underlying the proposed algorithm. These concepts are essential for understanding how the algorithm preserves the privacy of individual records while producing a meaningful summary of survival probabilities.

\subsection{Kaplan-Meier Estimation}

The \textbf{Kaplan-Meier estimator} is a non-parametric statistic used to estimate the survival function from time-to-event data, commonly used in medical research, reliability analysis, and other fields where survival analysis is required. For a set of individuals, each with an observed time-to-event or censoring time, the Kaplan-Meier estimator calculates the probability of survival beyond each observed time point, based on cumulative probabilities of surviving each individual time interval.

For a sequence of survival times \( \{t_1, t_2, \dots, t_n\} \), the Kaplan-Meier estimator \( S(t) \) at any time \( t_i \) is defined as:
\[
S(t_i) = \prod_{j=1}^{i} \left(1 - \frac{d_j}{n_j}\right),
\]
where:
\begin{itemize}
    \item \( d_j \) is the number of events (e.g., deaths or failures) at time \( t_j \),
    \item \( n_j \) is the number of individuals at risk just before time \( t_j \).
\end{itemize}

This estimator calculates survival probabilities by taking a product over time intervals, with each interval contributing to the overall probability of survival up to that time. The Kaplan-Meier curve, which plots \( S(t) \) over time, typically shows a stepwise decline, with each drop representing an event. This structure makes the Kaplan-Meier estimator cumulative by nature, where the probability at each time point depends on all previous events.

The Kaplan-Meier estimator assumes that the censoring is independent of the event of interest and that survival probabilities are constant within each time interval. However, survival probabilities can change rapidly in practice, especially when the risk of events increases or decreases significantly over time. This variability is particularly relevant to our algorithm, which aims to balance accuracy and privacy when estimating survival probabilities in a cumulative manner.

\subsection{Differential Privacy}

\textbf{Differential privacy} (DP) is a formal framework for preserving individual privacy when analyzing or releasing information about a dataset. The goal of differential privacy is to ensure that the output of an algorithm is statistically indistinguishable whether or not any single individual’s data is included in the dataset. This means that an adversary observing the released data cannot infer much about any particular individual’s record.

Mathematically, an algorithm \(\mathcal{A}\) is said to be \((\epsilon, \delta)\)-differentially private if, for any two datasets \(D\) and \(D'\) that differ by one individual, and for any possible output \(O\) of \(\mathcal{A}\),
\[
\Pr[\mathcal{A}(D) = O] \leq e^{\epsilon} \cdot \Pr[\mathcal{A}(D') = O] + \delta,
\]
where:
\begin{itemize}
    \item \( \epsilon \) (privacy budget) controls the strength of the privacy guarantee, with smaller values of \( \epsilon \) providing stronger privacy.
    \item \( \delta \) is a relaxation parameter that allows a small probability of the privacy guarantee being violated.
\end{itemize}

In this algorithm, we focus on \(\epsilon\)-differential privacy (with \(\delta = 0\)) and add \textbf{Laplace noise} to the survival probabilities to obscure the contribution of any single individual. The Laplace mechanism achieves \(\epsilon\)-differential privacy by adding noise calibrated to the sensitivity of the function, where sensitivity is the maximum amount that the function's output can change by adding or removing one individual from the dataset.

\subsubsection{Time-Indexed Noise for Kaplan-Meier Estimation}

Since the Kaplan-Meier estimator is cumulative, survival probabilities at early time points are more sensitive to individual records than those at later points, where fewer individuals remain at risk. This means that errors or noise introduced early in the timeline propagate through the entire survival curve, affecting all subsequent calculations. At early time points, the population at risk ($n_t$) is large, and while an individual’s contribution to the survival probability may seem small relative to this larger group, the impact of these contributions is amplified because of the cumulative nature of the Kaplan-Meier estimator. For example, a small reduction in survival probability at an early time point directly alters the base for calculating later probabilities, compounding its influence as the estimator progresses. As a result, early time points are considered globally sensitive, meaning that even small changes at these points have widespread effects across the entire curve.

In contrast, at later time points, the number of individuals remaining at risk is smaller. Each individual's contribution to the survival probability becomes more significant in relative terms—removing or adding one record can cause a noticeable change in the probability for that specific time point. However, these later probabilities are also smaller in magnitude, and their changes have a more localized impact because they do not propagate backward to earlier time points. This localized sensitivity at later time points is important but limited in scope. For example, if the survival probability at a late time point shifts slightly due to noise, it affects only the tail of the survival curve and not the broader structure that was established earlier.

The cumulative nature of the Kaplan-Meier estimator means that the survival probabilities at early time points act as a foundation for the entire curve. Even small distortions here can lead to large inaccuracies downstream due to the multiplicative structure of the estimator. While the fewer records at later time points increase their relative sensitivity, their impact on the overall survival curve is less critical compared to the amplified influence of errors at the start. To balance these effects, more noise is added at earlier time points to address their global sensitivity, while less noise is applied at later points to preserve the localized details of the tail. This dynamic approach ensures that the survival curve remains accurate and meaningful while protecting individual privacy.

To reflect this, we use \textbf{time-indexed noise scaling} in our algorithm, adding larger noise at earlier time points and reducing the noise over time. This approach preserves the cumulative nature of the Kaplan-Meier estimator while ensuring privacy protection that aligns with the varying sensitivity across time points.

\subsubsection{Challenges of Applying Differential Privacy to Kaplan-Meier Estimation}

Applying differential privacy to Kaplan-Meier estimation presents unique challenges. Adding noise to survival probabilities can lead to:
\begin{itemize}
    \item \textbf{Unrealistic Values}: Noise can push survival probabilities above 1 or below 0, which are not realistic for a survival curve.
    \item \textbf{Cumulative Noise Amplification}: Since Kaplan-Meier survival probabilities are computed as a product, noise introduced early in the time series compounds and affects later time points.
    \item \textbf{Non-Linear Declines}: Real survival curves often exhibit non-linear patterns, such as sharp drops or plateaus, which require careful handling of noise and privacy mechanisms.
\end{itemize}

To address these challenges, we introduce \textbf{dynamic clipping} to cap noisy survival probabilities within a realistic range and \textbf{smoothing} to reduce fluctuations caused by noise. Together, these techniques help maintain the shape and cumulative structure of the Kaplan-Meier curve, ensuring that the released differentially private estimator remains realistic and useful.

\subsection{Summary of the Approach}

The proposed algorithm combines the Kaplan-Meier estimator with differential privacy by using time-indexed noise, dynamic clipping, and smoothing. Time-indexed noise scaling ensures that earlier, more sensitive time points receive larger noise, while dynamic clipping prevents unrealistic values, and smoothing reduces noise variability. This approach provides a balance between privacy protection and accuracy in survival estimation, allowing the release of a meaningful survival curve that respects individual privacy.

\section{Differentially Private Kaplan-Meier Estimation Algorithm}

\subsection{Algorithm Design}
Our algorithm produces a differentially private Kaplan-Meier survival curve by adding time-indexed noise, applying dynamic clipping, and smoothing the survival probabilities to balance privacy and utility.

\begin{algorithm}[h]
\label{algo1}
\caption{Generalized Differentially Private Kaplan-Meier Estimation}
\begin{algorithmic}[1]
\Require Survival probabilities \(\{S(t_i)\}_{i=1}^n\), privacy budget \(\epsilon\), noise decay factor \(\alpha\), clipping thresholds \(\tau_{\text{start}}, \tau_{\text{end}}\), smoothing window \(w\)
\Ensure Differentially private and smoothed survival probabilities \(\{\tilde{S}_{\text{smoothed}}(t_i)\}_{i=1}^n\)
\For{each time index \(i = 1, \ldots, n\)}
    \State Set noise scale: \(\sigma_i = \frac{1}{\epsilon (1 + \alpha \cdot i)}\)
    \State Sample noise: \(N_i \sim \text{Laplace}(0, \sigma_i)\)
    \State Add noise to survival probability: \(\tilde{S}(t_i) = S(t_i) + N_i\)
    \State Set dynamic clipping threshold: \(\tau(i) = \tau_{\text{start}} - \left(\frac{i}{n}\right) \cdot (\tau_{\text{start}} - \tau_{\text{end}})\)
    \State Clip noisy probability: \(\tilde{S}_{\text{clipped}}(t_i) = \min(\tilde{S}(t_i), \tau(i))\)
\EndFor

\State Initialize empty list for smoothed probabilities \(\{\tilde{S}_{\text{smoothed}}(t_i)\}_{i=1}^n\)
\For{each time index \(i = 1, \ldots, n\)}
    \State Compute rolling mean over window \(w\): 
    \[
    \tilde{S}_{\text{smoothed}}(t_i) = \frac{1}{w} \sum_{j = \max(1, i - w/2)}^{\min(n, i + w/2)} \tilde{S}_{\text{clipped}}(t_j)
    \]
\EndFor

\State Apply cumulative minimum adjustment to ensure non-increasing survival probabilities:
\[
\tilde{S}_{\text{final}}(t_i) = \min_{j \leq i} \tilde{S}_{\text{smoothed}}(t_j)
\]

\Return \(\{\tilde{S}_{\text{final}}(t_i)\}_{i=1}^n\)
\end{algorithmic}
\end{algorithm}

\subsection{Mathematical Components of the Differentially Private Kaplan-Meier Estimator}

To achieve differential privacy and maintain the realism of the Kaplan-Meier curve, we introduce noise, dynamic clipping, and smoothing. Each step is carefully designed to reduce variance while preserving the survival curve's cumulative structure.

\subsubsection{Time-Indexed Noise Scaling}

To protect individual privacy, \textbf{Laplace noise}~\cite{dwork2014algorithmic} is added to each survival probability at every time point. The noise scale decreases over time to respect the natural decay of survival probabilities. 

For each survival probability \( S(t_i) \), the noise scale \( \sigma_i \) is defined as:
\[
\sigma_i = \frac{1}{\epsilon (1 + \alpha \cdot i)}
\]
where:
\begin{itemize}
    \item \( \epsilon \) controls the overall level of privacy,
    \item \( \alpha \) is a decay factor that determines how quickly the noise decreases over time.
\end{itemize}

This noise scaling produces a noisy estimate \(\tilde{S}(t_i)\) at each time point \( t_i \):
\[
\tilde{S}(t_i) = S(t_i) + \text{Laplace}\left(0, \sigma_i\right) = S(t_i) + \text{Laplace}\left(0, \frac{1}{\epsilon (1 + \alpha \cdot i)}\right)
\]
A larger value of \( \alpha \) results in faster noise decay over time, while a smaller value of \( \alpha \) maintains a more consistent level of noise across time points. This \textbf{time-indexed noise scaling} respects the cumulative structure of the Kaplan-Meier estimator.

\subsubsection{Dynamic Clipping to Maintain Realistic Probabilities}

Adding noise can result in extreme values for survival probabilities, particularly close to 0 or 1. \textbf{Dynamic clipping} is applied to cap these noisy values within a realistic range.

The clipping threshold \(\tau(i)\) is defined as:
\[
\tau(i) = \tau_{\text{start}} - \left(\frac{i}{n}\right) \cdot (\tau_{\text{start}} - \tau_{\text{end}})
\]
where:
\begin{itemize}
    \item \( \tau_{\text{start}} \) is the initial clipping threshold, typically near 1,
    \item \( \tau_{\text{end}} \) is the lower bound of the threshold, generally set around 0.5 or lower,
    \item \( n \) is the total number of time points.
\end{itemize}

The threshold decreases linearly over time, accommodating the typical decline in survival curves. By adjusting \(\tau_{\text{start}}\) and \(\tau_{\text{end}}\), we can control the initial and final bounds of the survival probabilities, maintaining a realistic curve shape.

\subsubsection{Adapting to Non-Linear Survival Curves}

To capture the shape of more complex survival curves, parameters can be tuned to adapt the algorithm to sharp drops or plateaus in survival probabilities:

\begin{itemize}
    \item \textbf{Noise Scaling (\(\alpha\))}: A larger \(\alpha\) can be used when survival probabilities decline sharply, while a smaller \(\alpha\) may be appropriate for gradual declines.
    \item \textbf{Clipping Thresholds (\(\tau_{\text{start}}\) and \(\tau_{\text{end}}\))}: For curves with rapid declines, a lower \(\tau_{\text{end}}\) avoids excessive clipping. For curves with plateaus, a higher \(\tau_{\text{end}}\) preserves realistic stability.
\end{itemize}

This flexibility allows the algorithm to handle various survival curve shapes while ensuring privacy.

\subsubsection{Smoothing the Noisy Survival Curve}

After adding noise and applying clipping, we apply a \textbf{smoothing operation} to improve the accuracy and stability of the noisy survival probabilities. This step reduces fluctuations caused by noise while preserving the cumulative nature of the Kaplan-Meier curve.

For a rolling window size \(w\), the smoothed survival probability at time \( t_i \) is defined as:
\[
\tilde{S}_{\text{smoothed}}(t_i) = \frac{1}{w} \sum_{j = i - w/2}^{i + w/2} \tilde{S}(t_j)
\]
The rolling mean reduces fluctuations while maintaining the curve's overall trend. 

Finally, to ensure that the Kaplan-Meier survival probabilities remain non-increasing, we apply a \textbf{cumulative minimum adjustment}:
\[
\tilde{S}_{\text{final}}(t_i) = \min_{j \leq i} \tilde{S}_{\text{smoothed}}(t_j)
\]
This preserves the property that survival probabilities should not increase over time, resulting in a more realistic differentially private Kaplan-Meier estimate.

\subsubsection{Summary of Parameter Choices and Adjustments}
By carefully selecting parameters for \textbf{noise scaling}, \textbf{clipping thresholds}, and \textbf{smoothing window size}, the algorithm can provide a privacy-preserving survival curve that accurately reflects various survival trends. These adjustments ensure the released Kaplan-Meier estimate is both useful and private, balancing differential privacy with practical utility.

\section{Theoretical Analysis}

\subsection{Privacy Analysis}
\begin{theorem}
    Let \(\tilde{S}_{\text{smoothed}}(t_i)\) be the differentially private Kaplan-Meier estimates produced by Algorithm~\ref{algo1} run with some parameters $\epsilon > 0$, $\alpha > 0$, and for $n$ total check-points. Then, we have that this Kaplan-Meier estimates satisfies $\hat\varepsilon$-DP for 
    \[
\hat \varepsilon = \epsilon \cdot (n+1)(\tfrac{\alpha n}{2}+ 1) = O(\epsilon \alpha n^2)\,.
\]
\end{theorem}
\begin{proof}
Recall that our computation of $\tilde S(t_i)$ is the following update:
\[
\tilde{S}(t_i) = S(t_i) + \text{Laplace}\left(0, \sigma_i\right) = S(t_i) + \text{Laplace}\left(0, \frac{1}{\epsilon (1 + \alpha \cdot i)}\right)
\]
The time-step $t_i$ (the check-in period) is public knowledge and does not need to br protected. Only the counts $d_j, n_j$ are private. Since, we have $ S(t_i) \in [0,1]$ the sensitivity of  $S(t_i)$ is at most 1. Thus, by the Laplace mechanism~\cite{dwork2014algorithmic}, the computation of each $\tilde{S}(t_i)$ is $\hat \varepsilon_i$-DP for 
\[
\hat \varepsilon_i = \epsilon (1 + \alpha \cdot i)\,.
\]
By composition, releasing all of the probabilities $(\tilde{S}(t_i))_{i=1,\dots,n}$ satisfies $\hat\varepsilon$-DP for
\[
\hat \varepsilon = \sum_{i=0}^n \hat \varepsilon_i = \sum_{i=0}^n \epsilon (1 + \alpha \cdot i) = \epsilon \cdot (n+1)(1 + \tfrac{\alpha n}{2})
\]
Finally, the rest of steps in Algorithm~\ref{algo1} such as clipping, smoothing, and the cumulative minimum adjustment can be seem as post-processing steps to the probabilities $(\tilde{S}(t_i))_{i=1,\dots,n}$. Hence, the entire Algorithm~\ref{algo1} satisfies $\hat\varepsilon$-DP for $\hat\varepsilon$ defined above.
\end{proof}

\subsection{Utility Analysis: Mean Squared Error (MSE) Bounds}
The utility of the algorithm is quantified by the mean squared error (MSE) between the true survival probabilities \(\{S(t_i)\}_{i=1}^n\) and the differentially private probabilities \(\{\tilde{S}(t_i)\}_{i=1}^n\):
\[
\text{MSE} = \frac{1}{n} \sum_{i=1}^n (S(t_i) - \tilde{S}(t_i))^2
\]
As \(\epsilon\) increases, MSE decreases due to reduced noise. The proposed algorithm's adaptive noise scaling and dynamic clipping maintain this MSE below practical thresholds, achieving a balance between privacy and utility.

\subsection{Practical Utility Loss Bound}
\begin{theorem}
    
Let \(\tilde{S}_{\text{smoothed}}(t_i)\) be the differentially private Kaplan-Meier estimate of the true survival probability \(S(t_i)\) at time \(t_i\), produced by adding time-indexed Laplace noise to each \(S(t_i)\), applying a dynamic clipping threshold, and smoothing the noisy estimates. Then, with probability at least \(1 - \delta\), the Mean Squared Error (MSE) between \(\{S(t_i)\}_{i=1}^n\) and \(\{\tilde{S}_{\text{smoothed}}(t_i)\}_{i=1}^n\) is bounded by:
\[
\text{MSE} \leq O\left(\frac{\alpha n^3}{\hat\varepsilon}\right) + O\left(\frac{\log(1/\delta)}{n}\right) + O(C^2) + O(S^2)
\]
where:
\begin{itemize}
    \item $\hat \varepsilon$ is the total privacy budget
    \item \(n\) is the number of time points,
    \item \(\alpha\) is a scaling factor for time-indexed noise,
    \item \(\delta\) is the probability tolerance for exceeding this bound,
    \item \(C\) is the maximum bias introduced by dynamic clipping, and
    \item \(S\) is the maximum bias introduced by smoothing.
\end{itemize}
\end{theorem} 
\begin{proof}
We break down the proof into the effects of noise, clipping, and smoothing, and then combine them with a high-probability bound using Bernstein's inequality.

\textbf{Cumulative Noise Impact:}

Each Kaplan-Meier estimate \(S(t_i)\) is perturbed by adding Laplace noise with a time-indexed scale \(\sigma_i = \frac{1}{\epsilon (1 + \alpha i)}\). The noisy estimate \(\tilde{S}(t_i)\) at each time \(t_i\) is given by:
\[
\tilde{S}(t_i) = S(t_i) + N_i, \quad N_i \sim \text{Laplace}\left(0, \sigma_i\right).
\]

Due to the cumulative nature of the Kaplan-Meier estimator, noise introduced at earlier time points propagates forward, accumulating over time. The expected Mean Squared Error (MSE) due to noise is:
\[
\text{Noise MSE} = \frac{1}{n} \sum_{i=1}^n \mathbb{E}[(\tilde{S}(t_i) - S(t_i))^2].
\]

Substituting \(\tilde{S}(t_i) - S(t_i) = N_i\), we have:
\[
\text{Noise MSE} = \frac{1}{n} \sum_{i=1}^n \mathbb{E}[N_i^2].
\]
The variance of Laplace noise is \(\text{Var}[N_i] = 2\sigma_i^2 = \frac{2}{\epsilon^2 (1 + \alpha i)^2}\), so:
\[
\text{Noise MSE} = \frac{1}{n} \sum_{i=1}^n \frac{2}{\epsilon^2 (1 + \alpha i)^2}.
\]

Approximating the summation term and using \(\hat\varepsilon = O(\epsilon n^2 \alpha)\), we derive:
\[
\text{Noise MSE} \leq O\left(\frac{\alpha n^3}{\hat\varepsilon}\right).
\]

This term reflects the cumulative impact of noise over time.

\textbf{Cumulative Clipping Impact:}

Dynamic clipping caps the noisy survival probabilities \(\tilde{S}(t_i)\) at each time point using a time-dependent threshold \(\tau(i)\), defined as:
\[
\tau(i) = \tau_{\text{start}} - \left(\frac{i}{n}\right) \cdot (\tau_{\text{start}} - \tau_{\text{end}}).
\]

The clipped estimate is:
\[
\tilde{S}_{\text{clipped}}(t_i) = \min(\tilde{S}(t_i), \tau(i)).
\]

Let \(C\) be the maximum bias introduced by clipping:
\[
C = \max_{i} |\tilde{S}(t_i) - \tilde{S}_{\text{clipped}}(t_i)|.
\]

The clipping bias accumulates over time due to the recursive nature of the Kaplan-Meier estimator. The cumulative impact of clipping adds an \(O(C^2)\) term to the MSE:
\[
\text{Clipping MSE} \approx O(C^2).
\]

\textbf{Cumulative Smoothing Impact:}

Smoothing reduces random fluctuations by averaging survival probabilities within a window \(w\). The smoothed survival probability at time \(t_i\) is:
\[
\tilde{S}_{\text{smoothed}}(t_i) = \frac{1}{w} \sum_{j = i - w/2}^{i + w/2} \tilde{S}_{\text{clipped}}(t_j).
\]

However, smoothing introduces a bias \(S\), particularly in regions where survival rates change sharply. Let \(S\) be the maximum smoothing bias:
\[
S = \max_{i} \left| \tilde{S}_{\text{clipped}}(t_i) - \tilde{S}_{\text{smoothed}}(t_i) \right|.
\]

The cumulative impact of smoothing contributes an \(O(S^2)\) term to the MSE:
\[
\text{Smoothing MSE} \approx O(S^2).
\]

\textbf{High-Probability Bound Using Bernstein's Inequality~\cite{Vershynin_2018}:}

The noise term \(\frac{1}{n} \sum_{i=1}^n N_i^2\) is a sum of independent random variables. Let:
\[
Z = \sum_{i=1}^n N_i^2, \quad \text{where } N_i \sim \text{Laplace}(0, \sigma_i).
\]

Using Bernstein's inequality for bounded random variables, we have:
\[
\Pr\left[Z - \mathbb{E}[Z] \geq t\right] \leq \exp\left(-\frac{t^2}{2 \text{Var}[Z] + M t/3}\right),
\]
where \(M = \max_i |N_i^2|\).

\textbf{Expectation of \(Z\):}
\[
\mathbb{E}[Z] = \sum_{i=1}^n \mathbb{E}[N_i^2] = \sum_{i=1}^n \frac{2}{\epsilon^2 (1 + \alpha i)^2}.
\]

\textbf{Variance of \(Z\):}
\[
\text{Var}[Z] = \sum_{i=1}^n \text{Var}[N_i^2].
\]

The variance of \(N_i^2\) is bounded as:
\[
\text{Var}[N_i^2] = O\left(\frac{1}{\epsilon^4 (1 + \alpha i)^4}\right).
\]

Thus:
\[
\text{Var}[Z] = \sum_{i=1}^n O\left(\frac{1}{\epsilon^4 (1 + \alpha i)^4}\right).
\]

For \(t = O(\log(1/\delta))\), the probability of deviation decays exponentially:
\[
\Pr\left[Z \geq \mathbb{E}[Z] + t\right] \leq \exp\left(-\frac{t^2}{2 \text{Var}[Z] + M t/3}\right).
\]

\textbf{High-Probability MSE Bound:}
Rescaling by \(\frac{1}{n}\), the contribution of noise to the MSE is:
\[
\text{MSE}_\text{noise} \leq O\left(\frac{\alpha n^3}{\hat\varepsilon}\right) + O\left(\frac{\log(1/\delta)}{n}\right).
\]

Adding the contributions from noise, clipping, and smoothing, the total MSE is:
\[
\text{MSE} \leq O\left(\frac{\alpha n^3}{\hat\varepsilon}\right) + O(C^2) + O(S^2) + O\left(\frac{\log(1/\delta)}{n}\right).
\]

With probability at least \(1 - \delta\), the MSE satisfies:
\[
\Pr\left[\text{MSE} > O\left(\frac{\alpha n^3}{\hat\varepsilon}\right) + O(C^2) + O(S^2) + O\left(\frac{\log(1/\delta)}{n}\right)\right] \leq \delta.
\]
\end{proof}

\section{Experiments}

\subsection{Dataset}
For our experiments, we use the NCCTG lung cancer dataset~\cite{R_survival_lung}~\cite{loprinzi1994prospective}, which is publicly available in the R survival package. The dataset consists of 228 observations of patients with lung cancer, along with 10 variables. These variables include survival time, event status (whether the event of interest, i.e., death, occurred), and other clinical features such as age, gender, and treatment group. The primary focus of our experiment is on the time-to-event data (survival times) and the event status (whether or not the event was observed).

The dataset is used to evaluate the performance of our differentially private Kaplan-Meier estimator in preserving privacy while maintaining the utility of survival probability estimates. Specifically, we focus on the survival times and event status, which are the key variables for survival analysis.

\subsection{Goals of the Experiment}
The main objectives of our experiments are as follows:

\begin{itemize}
    \item \textbf{Evaluate the accuracy-privacy trade-off}: We aim to assess how different levels of privacy, represented by the privacy budget $\epsilon$, affect the utility of the Kaplan-Meier survival estimates. We do this by comparing the differentially private Kaplan-Meier estimates to the non-private estimates in terms of \textbf{Root Mean Squared Error (RMSE)}.
    
    \item \textbf{Assess the impact of differential privacy mechanisms}: We want to examine the effects of our privacy-preserving mechanisms (time-indexed Laplace noise scaling, dynamic clipping, and smoothing) on the survival estimates. Specifically, we will investigate how these mechanisms impact the accuracy of survival estimates while preserving individual privacy.
    
    \item \textbf{Test robustness under various privacy budgets}: By varying the privacy budget $\epsilon$, we evaluate how the privacy mechanism behaves under different levels of noise, with a focus on understanding the relationship between $\epsilon$ and the trade-off between privacy and utility.
    
    \item \textbf{Assess visual fidelity}: We will compare the non-private Kaplan-Meier curve to the differentially private Kaplan-Meier estimates visually to ensure that the survival curve retains its natural shape, even under privacy-preserving conditions.
    
    \item \textbf{Evaluate membership inference attack resistance}: We aim to assess whether the differentially private Kaplan-Meier estimator effectively mitigates membership inference attacks, which attempt to determine if a particular data point was part of the training dataset. We do this by evaluating the privacy leakage of the estimator at different privacy budgets.
\end{itemize}

\subsection{Experimental Setup}
The experiment involves the following key steps:

\begin{enumerate}
    \item \textbf{Non-Private Kaplan-Meier Estimation}: We start by fitting a non-private Kaplan-Meier estimator using the survival times and event status data from the NCCTG lung cancer dataset. This will serve as the baseline for comparison with the differentially private estimates.
    
    \item \textbf{Differentially Private Kaplan-Meier Estimation}: For each privacy budget $\epsilon$, we apply the differential privacy mechanisms (Laplace noise scaling, dynamic clipping, and smoothing) to the Kaplan-Meier estimate.
    \begin{itemize}
        \item We use time-indexed Laplace noise that scales according to the time index to provide differential privacy while accounting for the varying sensitivities of survival probabilities at different time points.
        \item Dynamic clipping is applied to adjust the clipping threshold at each time point, ensuring that extreme values do not distort the survival curve.
        \item Smoothing is applied using a rolling window to reduce the impact of noise fluctuations and maintain the natural shape of the survival curve.
    \end{itemize}
    
    \item \textbf{Evaluation Metrics}: We compute the Root Mean Squared Error (RMSE) between the non-private Kaplan-Meier estimate and the differentially private estimates to assess how well the private estimates approximate the true survival function.
    
    \item \textbf{Privacy Performance}: We measure how effectively the differentially private algorithm prevents membership inference attacks and analyze the resistance of the model to privacy leakage at different privacy budgets.

\end{enumerate}

\subsection{Parameters for Differential Privacy Mechanisms}

To evaluate the privacy-utility trade-off in our differential privacy mechanism, we experiment with a range of privacy budget values, $\epsilon$, spanning from low (stronger privacy) to high (weaker privacy). Additionally, we test various settings for time-indexed Laplace noise scaling ($\alpha$), dynamic clipping thresholds ($\tau_{\text{start}}$, $\tau_{\text{end}}$), and smoothing window size ($w$). These parameters are tuned based on dataset characteristics and prior experimentation to optimize utility while maintaining privacy. Table~\ref{tab:dp_params} summarizes the key parameters used in the differential privacy mechanisms for Kaplan-Meier estimation.

\begin{table}[ht]
\centering
\caption{Parameters for Differentially Private Kaplan-Meier Estimation}
\begin{tabular}{|l|l|}
\hline
\textbf{Parameter} & \textbf{Description and Typical Values} \\ \hline
$\epsilon$ & Parameter controlling privacy budget) \\ \hline
$\alpha$ & Time-indexed scaling factor for Laplace noise (default: 0.05) \\ \hline
$\hat \varepsilon = \epsilon \cdot (n+1)(\tfrac{\alpha n}{2}+ 1)$ & Actual total privacy budget \\ \hline
$\tau_{\text{start}}$ & Initial clipping threshold for survival probabilities (default: 0.95) \\ \hline
$\tau_{\text{end}}$ & Final clipping threshold for survival probabilities (default: 0.5) \\ \hline
$w$ & Smoothing window size for rolling mean (default: 3) \\ \hline
\end{tabular}
\label{tab:dp_params}
\end{table}

\section{Results}

\subsection{Privacy vs Utility}

The sensitivity analysis was conducted to understand how the privacy budget $\epsilon$ impacts the accuracy of the differentially private (DP) Kaplan-Meier survival estimates. Since privacy-preserving algorithms introduce noise to protect individual data points, there is a trade-off between privacy and utility. A higher privacy budget (larger $\epsilon$) permits more accurate estimates by adding less noise, whereas a lower privacy budget enforces stronger privacy by introducing more noise. This analysis seeks to quantify this trade-off by examining how RMSE, a measure of error between DP and non-private survival estimates, varies with different $\epsilon$ values.

\subsubsection{Impact of Low Privacy Budgets}

At lower values of $\epsilon$ (e.g., 0.1 and 1.0), the RMSE remains relatively high. This indicates greater distortion in the survival probabilities due to the stronger privacy constraints. For instance, when $\epsilon = 0.1$, the RMSE is approximately 0.57, which reflects a high level of noise-induced error in the survival estimates. These findings suggest that with lower privacy budgets, the added noise more significantly impacts the accuracy of the DP survival estimates.

\subsubsection{Improvement with Higher Privacy Budgets}

As $\epsilon$ increases to values such as 4.0 and beyond, the RMSE decreases substantially, reaching around 0.04 at $\epsilon = 10$. This reduction in RMSE at higher $\epsilon$ values suggests that allowing more privacy budget leads to better accuracy, bringing the DP survival probabilities closer to their non-private counterparts. 

\subsubsection{Summary of Privacy-Utility Trade-Off}

This analysis emphasizes the \textbf{trade-off between privacy and utility}: larger privacy budgets (less privacy) result in more accurate survival estimates, while stricter privacy constraints introduce greater noise, leading to increased RMSE. Figure~\ref{fig:sensitivity-analysis} visually depicts this trade-off, demonstrating that with higher privacy budgets, the survival estimates more closely approximate non-private estimates, achieving improved utility.

\begin{figure}[h]
    \centering
    \includegraphics[width=0.7\textwidth]{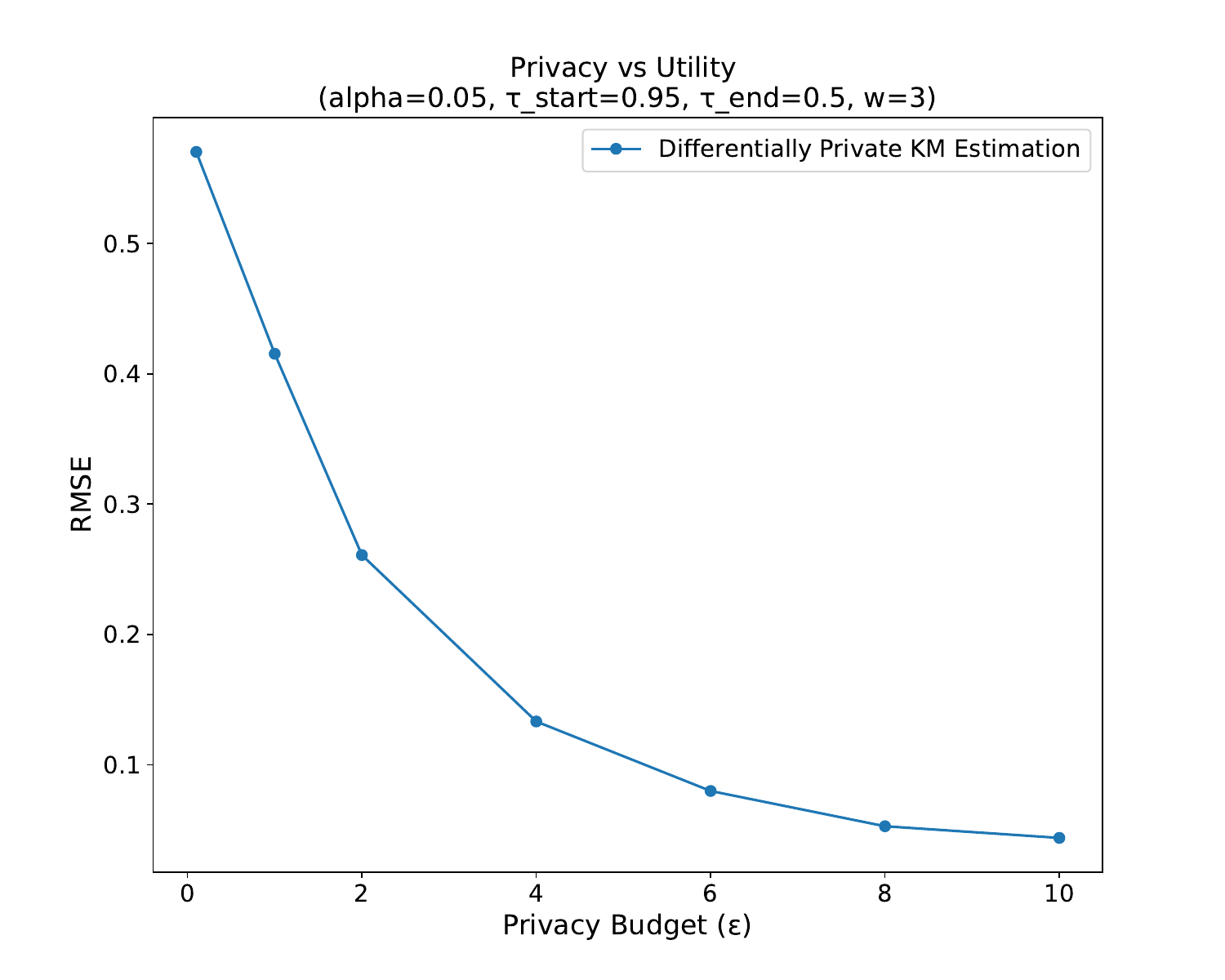}
    \caption{Sensitivity Analysis of RMSE Across Privacy Budgets ($\epsilon$) Showing the Privacy-Utility Trade-Off.}
    \label{fig:sensitivity-analysis}
\end{figure}

\subsection{Differentially Private KM Curves Comparison}

The comparison between the non-private and differentially private (DP) Kaplan-Meier (KM) survival curves demonstrates how varying the privacy budget, $\epsilon$, influences the accuracy and stability of survival probability estimates over time (refer Figure~\ref{fig:km-comparison}). At lower values of $\epsilon$, such as $\epsilon = 0.1$, the DP KM survival probabilities show significant fluctuations, with probabilities dropping to zero at various points. For instance, beginning from time $t = 5.0$, the DP KM survival probabilities remain at zero across several time points, diverging markedly from the non-private estimates. This indicates that when the privacy budget is very low, the noise added for privacy preservation severely distorts the survival probabilities.

As $\epsilon$ increases, the DP KM curves begin to stabilize and align more closely with the non-private survival estimates. For example, with $\epsilon = 1$, the DP KM curve maintains positive survival probabilities but still shows discrepancies from the non-private curve at certain time points, particularly at earlier times. At $\epsilon = 2$, the DP KM curve shows even fewer fluctuations and provides more stable estimates than at $\epsilon = 1$, although it remains somewhat lower than the non-private survival probabilities.

With higher values of $\epsilon$, such as $\epsilon = 4$, $\epsilon = 6$, and $\epsilon = 8$, the DP KM estimates continue to improve in alignment with the non-private survival probabilities, maintaining consistency without abrupt drops. By the time $\epsilon = 10$, the DP KM curve closely approximates the non-private survival probabilities across all observed time points, producing survival estimates that nearly match the non-private curve.

This pattern is consistent with findings from the RMSE sensitivity analysis, where higher values of $\epsilon$ yield lower RMSE values, indicating enhanced utility. With a larger privacy budget, the impact of noise diminishes, enabling the DP KM curves to better reflect the true survival probabilities.

In summary: \begin{itemize} \item \textbf{Low privacy budgets} (e.g., $\epsilon = 0.1$) introduce substantial noise, resulting in significant fluctuations and frequent zero probabilities. This demonstrates that stringent privacy constraints severely impact the utility of survival estimates. \item \textbf{Moderate privacy budgets} (e.g., $\epsilon = 4$, $\epsilon = 6$, and $\epsilon = 8$) allow the DP KM curves to align more closely with the non-private survival probabilities, maintaining smooth estimates with only minor deviations. \item \textbf{High privacy budgets} (e.g., $\epsilon = 10$) provide DP KM estimates that almost entirely match the non-private survival probabilities, suggesting that the added noise has minimal impact on the utility of the survival estimates. \end{itemize}

\begin{figure}[h] 
\centering 
\includegraphics[width=0.7\textwidth]{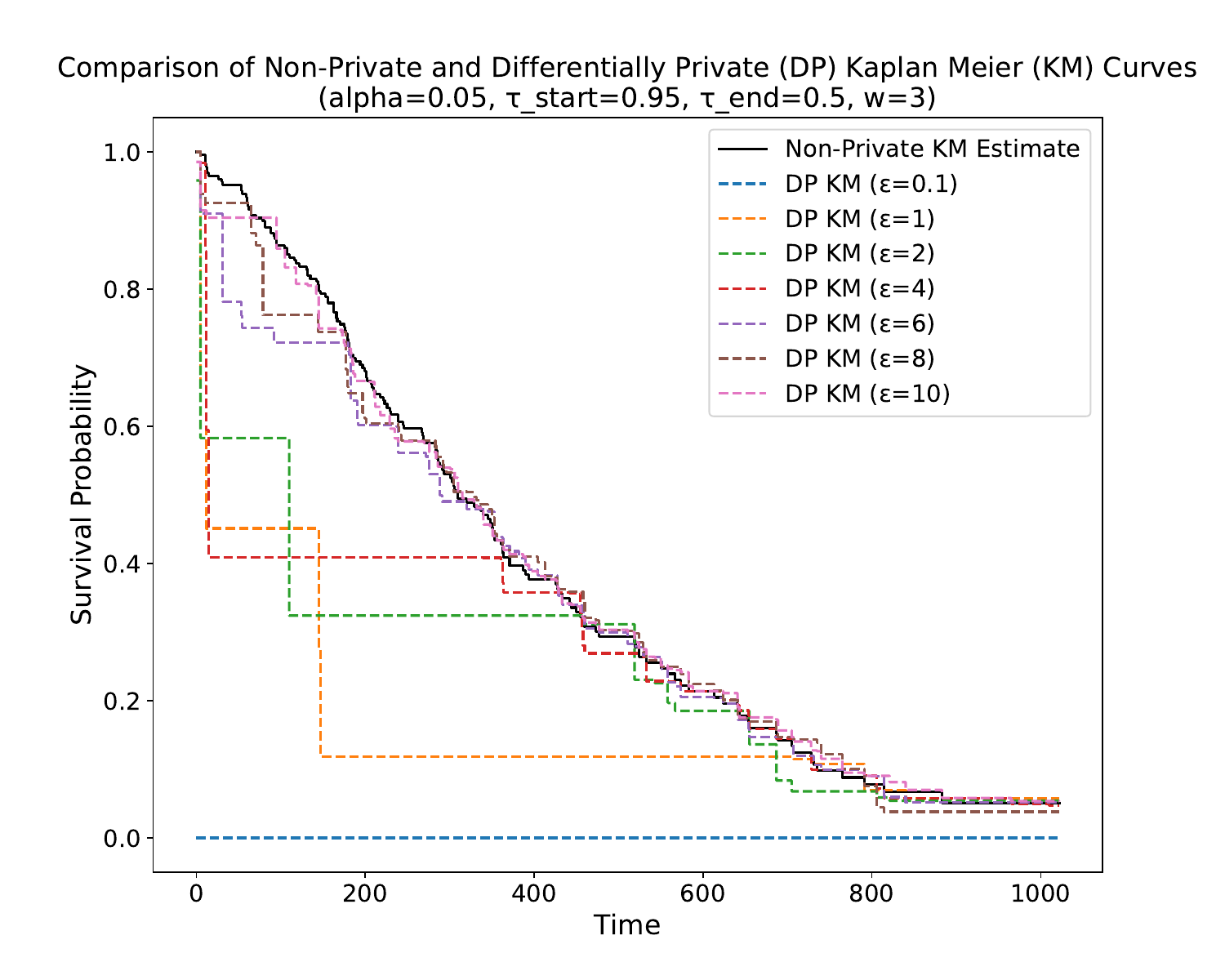} 
\caption{Comparison of Non-Private and Differentially Private Kaplan-Meier Curves Across Varying Privacy Budgets ($\epsilon$).} 
\label{fig:km-comparison}
\end{figure}

Both the privacy-utility analysis (Figure~\ref{fig:sensitivity-analysis}) and the DP KM curve comparisons (Figure~\ref{fig:km-comparison}) illustrate that increasing $\epsilon$ enhances the alignment between DP and non-private survival estimates. This is demonstrated by decreasing RMSE and increasingly stable DP KM curves. While low privacy budgets cause significant deviations from non-private estimates, particularly at later time points with fewer observations, higher privacy budgets effectively mitigate the noise, leading to more reliable and accurate survival probability estimates.

In summary:
\begin{itemize}
    \item \textbf{Low $\epsilon$ values} yield higher noise, higher RMSE, and fluctuating survival estimates, emphasizing privacy at the cost of utility.
    \item \textbf{High $\epsilon$ values} result in low RMSE and DP KM curves that closely align with non-private estimates, achieving better utility.
\end{itemize}

\subsection{Time-Indexed Noise Scaling Effect}
\begin{figure}[htbp!]
    \centering
    \includegraphics[width=0.8\textwidth]{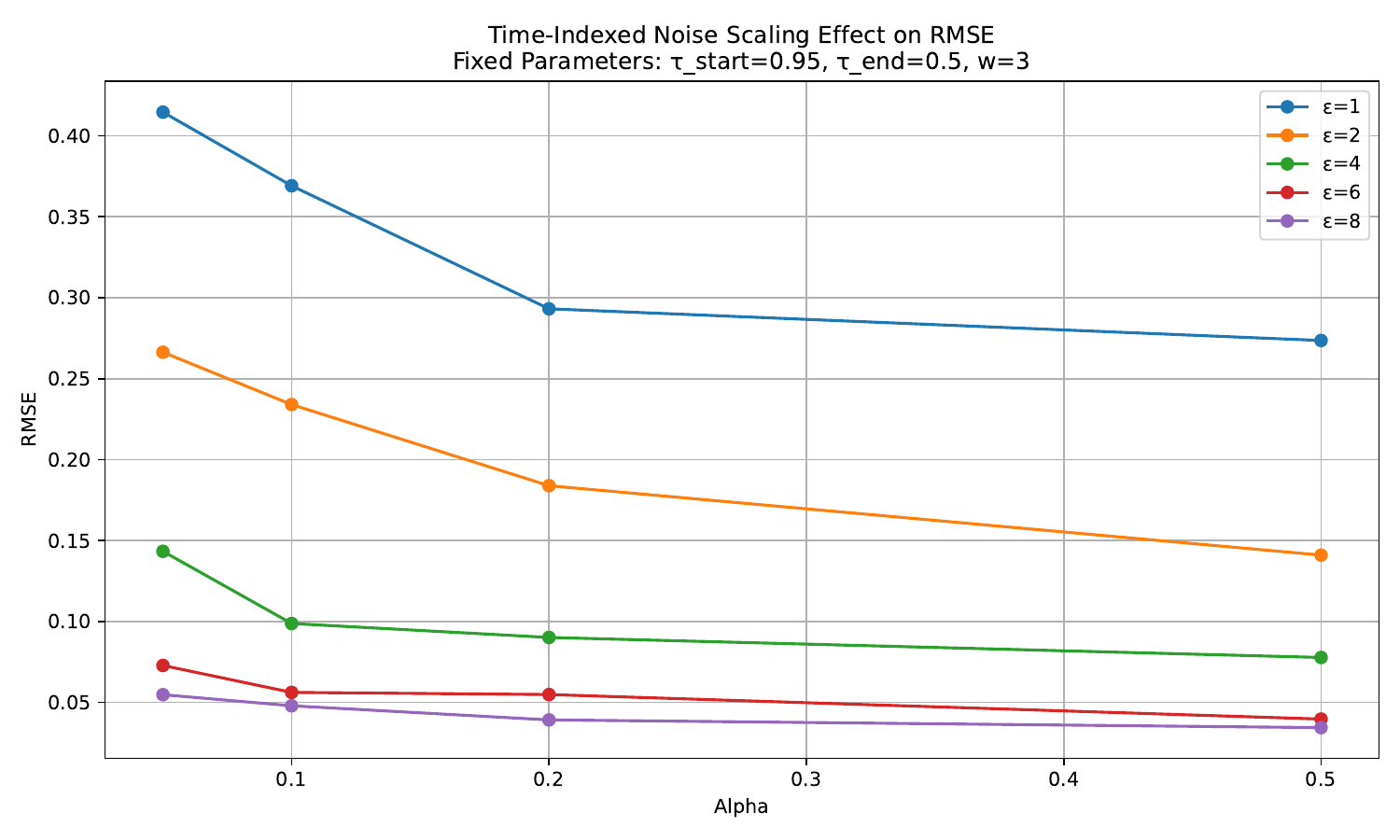}
    \caption{Effect of Time-Indexed Noise Scaling on RMSE for varying $\alpha$ values under different $\epsilon$.}
    \label{fig:time-indexed-noise-scaling}
\end{figure}
In this experiment, we evaluate the effect of varying the noise scaling factor $\alpha$ on RMSE under different privacy budgets $\epsilon$. The noise scaling factor $\alpha$ controls how quickly the noise decays over time. Lower $\alpha$ values indicate slower decay, resulting in higher noise retention at later time points.

The results, shown in Table~\ref{table:time-indexed-noise-scaling} and Figure~\ref{fig:time-indexed-noise-scaling}, indicate that as $\alpha$ increases, the \textbf{mean RMSE} generally decreases, particularly for higher values of $\epsilon$. This is because higher $\alpha$ values cause the noise to decay faster, reducing its impact on later time points. Additionally, confidence intervals show narrower bounds at higher $\epsilon$, reflecting greater accuracy in survival probability estimates with a larger privacy budget.

For instance, at $\epsilon = 8$ and $\alpha = 0.50$, the mean RMSE is minimized at \textbf{0.039} with a 95\% confidence interval of [0.0075, 0.160]. This shows that the algorithm provides a strong balance between privacy and utility, particularly with higher privacy budgets.

\begin{table}[h]
    \centering
    \caption{Selected results illustrating the effect of time indexed noise scaling confidence intervals are presented in this table for an illustrative purpose. The complete set of results can be found in Figure~\ref{fig:time-indexed-noise-scaling} }
    \label{table:time-indexed-noise-scaling}
    \begin{tabular}{cccccc}
        \toprule
        $\epsilon$ & $\alpha$ & Mean RMSE & CI Lower & CI Upper \\
        \midrule
        1 & 0.05 & 0.4257 & 0.2289 & 0.5990 \\
        1 & 0.10 & 0.3825 & 0.1701 & 0.5990 \\
        1 & 0.20 & 0.3074 & 0.0997 & 0.5990 \\
        8 & 0.50 & 0.0391 & 0.0075 & 0.1601 \\
        \bottomrule
    \end{tabular}
\end{table}

\subsection{Clipping Effect}
\begin{figure}[htbp!]
    \centering
    \includegraphics[width=0.8\textwidth]{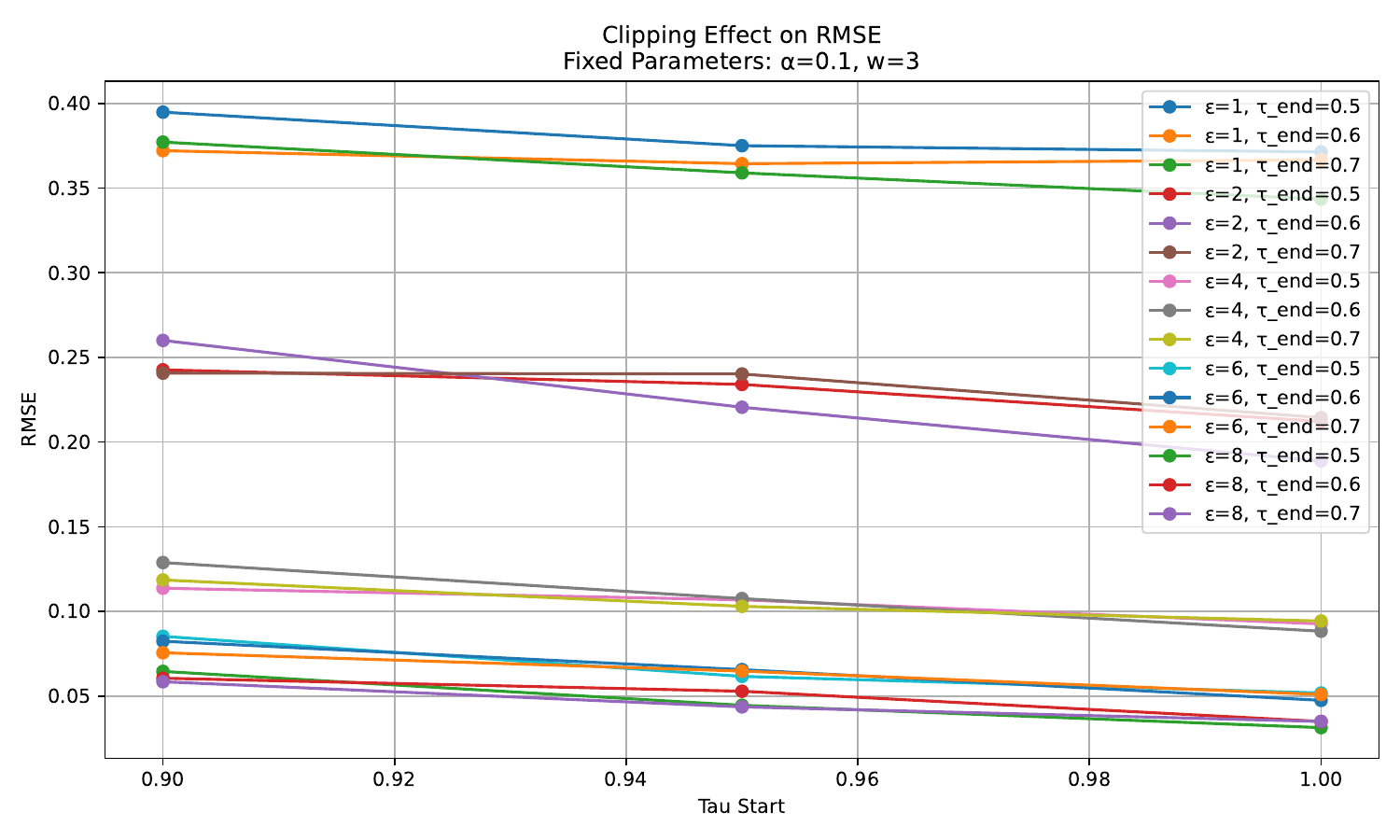}
    \caption{Effect of Clipping on RMSE for varying $\tau_{\text{start}}$ and $\tau_{\text{end}}$ values under different $\epsilon$.}
    \label{fig:clipping-effect}
\end{figure}
The clipping thresholds ($\tau_{\text{start}}$ and $\tau_{\text{end}}$) control the upper bounds on noisy survival probabilities, limiting extreme noise effects. We explore a range of starting and ending thresholds for different values of $\epsilon$ to understand how this impacts RMSE.

Table~\ref{table:clipping-effect} and Figure~\ref{fig:clipping-effect} present the results, showing that lower clipping thresholds (e.g., $\tau_{\text{start}} = 0.90$, $\tau_{\text{end}} = 0.5$) generally yield higher RMSE, as they restrict the survival probabilities more strictly. Higher values, like $\tau_{\text{start}} = 1.00$ and $\tau_{\text{end}} = 0.7$, provide lower RMSE with broader confidence intervals, indicating some variability but better alignment with actual survival probability values.

For example, for $\epsilon = 8$ with $\tau_{\text{start}} = 1.00$ and $\tau_{\text{end}} = 0.6$, the mean RMSE was \textbf{0.0347} with a 95\% confidence interval of [0.0128, 0.100], which suggests that clipping thresholds can be fine-tuned to improve estimation accuracy while preserving privacy.

\begin{table}[h]
    \centering
    \caption{Selected results illustrating the effect of clipping with confidence intervals are presented in this table for an illustrative purpose. The complete set of results can be found in Figure~\ref{fig:clipping-effect}.}
    \label{table:clipping-effect}
    \begin{tabular}{cccccc}
        \toprule
        $\epsilon$ & $\tau_{\text{start}}$ & $\tau_{\text{end}}$ & Mean RMSE & CI Lower & CI Upper \\
        \midrule
        1 & 0.90 & 0.5 & 0.3845 & 0.1732 & 0.5990 \\
        1 & 0.95 & 0.5 & 0.3516 & 0.1588 & 0.5990 \\
        1 & 1.00 & 0.7 & 0.3760 & 0.1605 & 0.5990 \\
        8 & 1.00 & 0.6 & 0.0347 & 0.0128 & 0.1002 \\
        \bottomrule
    \end{tabular}
\end{table}

\subsection{Smoothing Effect}
\begin{figure}[htbp!]
    \centering
    \includegraphics[width=0.8\textwidth]{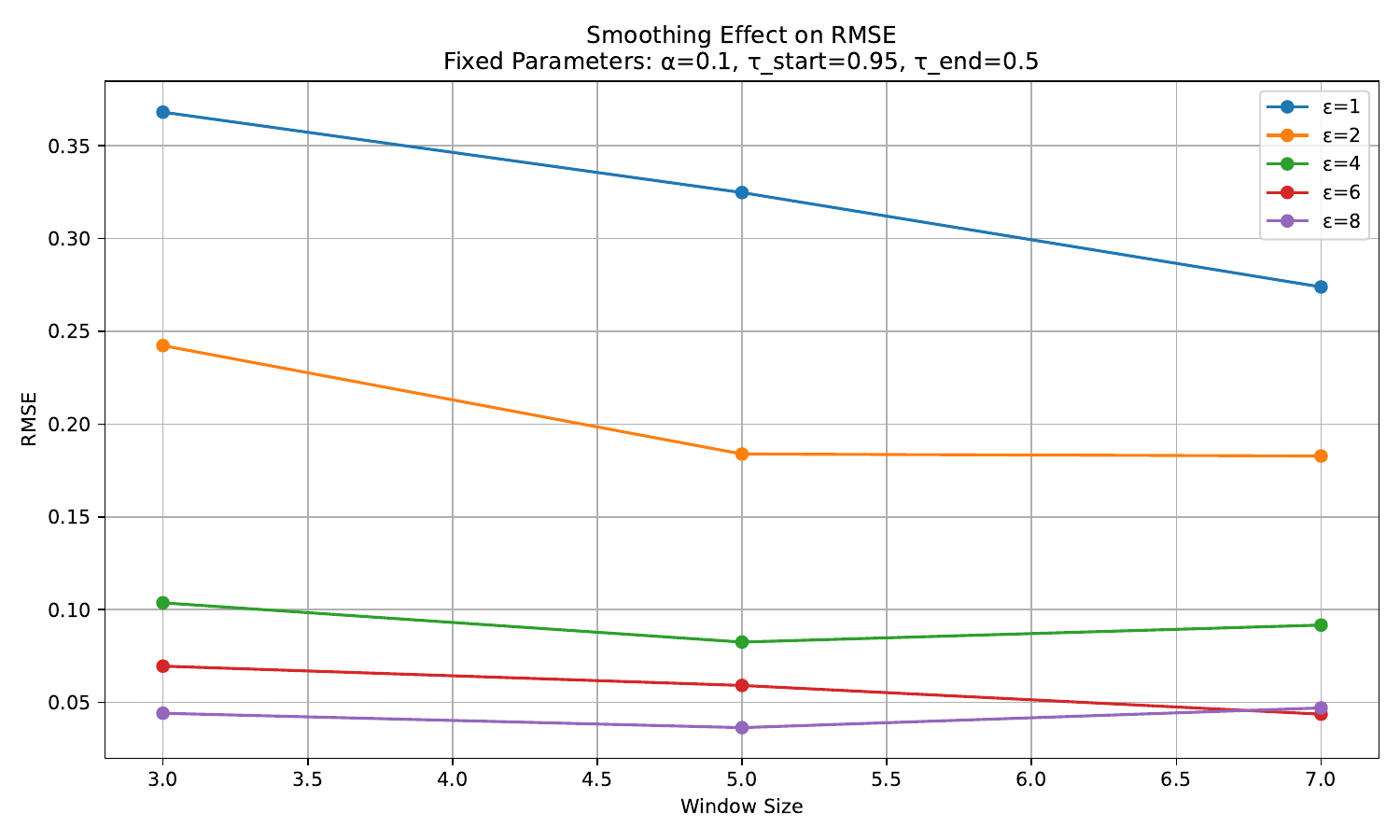}
    \caption{Effect of Smoothing on RMSE for varying window sizes $w$ under different $\epsilon$.}
    \label{fig:smoothing-effect}
\end{figure}
Smoothing, achieved by a moving average window, mitigates abrupt changes in noisy survival probabilities, enhancing the accuracy of the cumulative survival curve. This experiment examines the impact of window sizes $w = 3, 5,$ and $7$ on RMSE across different $\epsilon$ values.

Table~\ref{table:smoothing-effect} and Figure~\ref{fig:smoothing-effect} summarize the results. A larger smoothing window tends to decrease RMSE, especially for lower privacy budgets (e.g., $\epsilon = 1$ and $\epsilon = 2$), as it smooths out fluctuations more effectively. However, for high privacy budgets (e.g., $\epsilon = 8$), the improvement in RMSE with larger windows is less pronounced. This indicates that smoothing is most beneficial in scenarios with tighter privacy constraints where noise levels are inherently higher.

For example, with $\epsilon = 4$ and $w = 7$, the mean RMSE was \textbf{0.095} with a confidence interval of [0.0191, 0.3936], illustrating how smoothing can stabilize survival estimates even with moderate privacy budgets.

\begin{table}[h]
    \centering
    \caption{Selected Results of Smoothing Effect Confidence Intervals are presented in this table for an illustrative purpose. The complete set of results can be found in Figure~\ref{fig:smoothing-effect}.}
    \label{table:smoothing-effect}
    \begin{tabular}{ccccc}
        \toprule
        $\epsilon$ & Window Size $w$ & Mean RMSE & CI Lower & CI Upper \\
        \midrule
        1 & 3 & 0.3635 & 0.1577 & 0.5990 \\
        1 & 5 & 0.3375 & 0.1164 & 0.5990 \\
        8 & 7 & 0.0359 & 0.0164 & 0.1113 \\
        \bottomrule
    \end{tabular}
\end{table}

\subsection{Membership Inference Attacks}
\begin{figure}[htbp!]
    \centering
    \includegraphics[width=16cm, height=20cm]{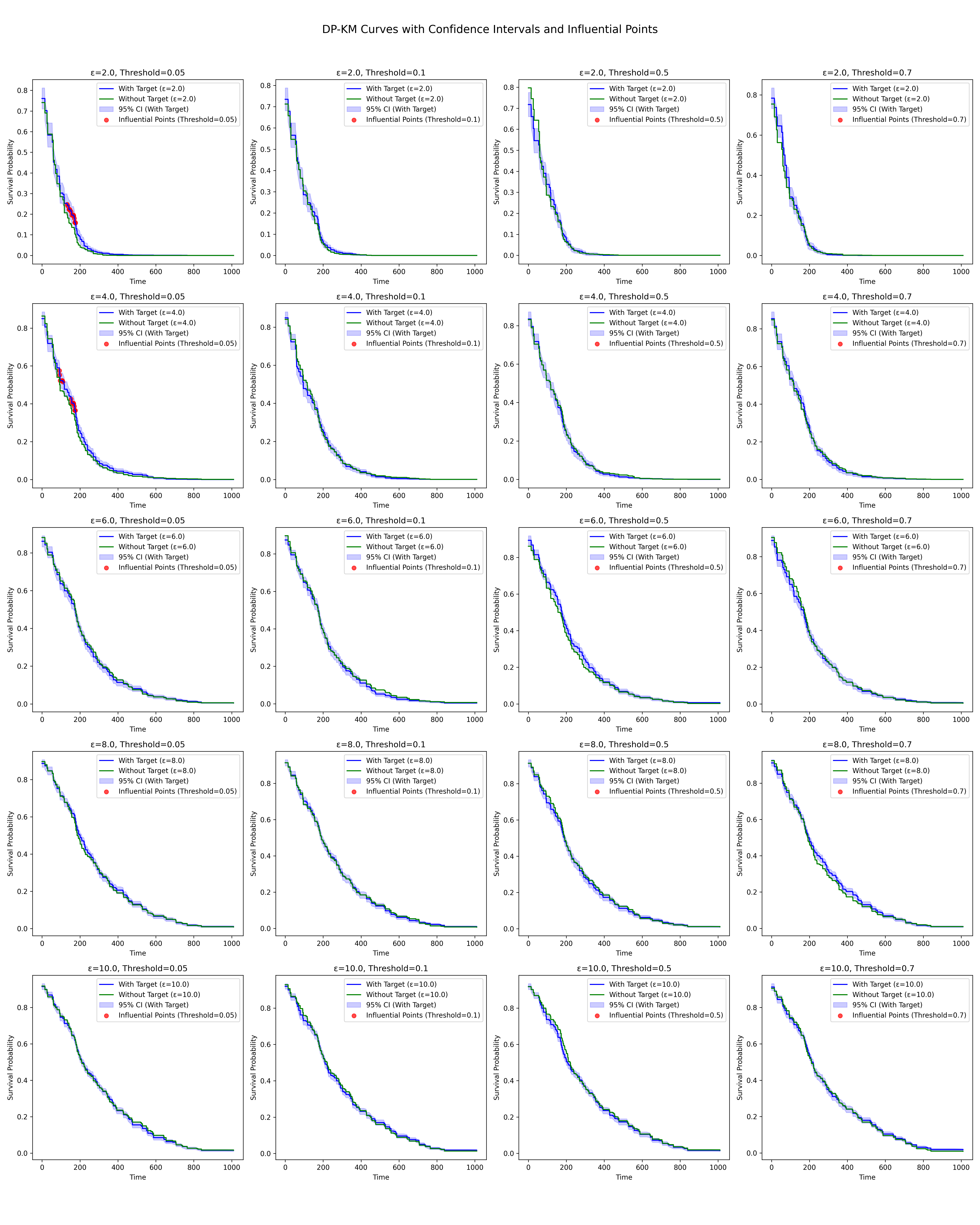}
    \caption{Membership Inference Attack Analysis: Influential Points Across Privacy Budgets and Thresholds for Kaplan-Meier Curves.}
    \label{fig:membership-inference}
\end{figure}

A membership inference attack aims to determine if a specific individual's data point was included in the dataset used to generate a model. In our experiments, we conduct membership inference attacks on the differentially private Kaplan-Meier (DP-KM) survival curves by examining the effects of excluding a targeted data point from the dataset. We perform the attack by comparing the absolute change in magnitude of the survival probabilities from models generated with and without the target data point. We leave the more sophisticated approaches of computing an empirical lower-bounds to the privacy budget using membership inference attacks (such as in~\cite{jagielski2020auditing,nasr2021adversary}) for future work.

To quantify the impact of excluding the target point, we calculate the differences in survival probabilities between the two models across multiple trials. We then identify \textit{influential points} as time points where the absolute difference in survival probability between the two models exceeds a pre-defined threshold. In our analysis, we vary the threshold values across 0.05, 0.1, 0.5, and 0.7 to observe the effects of privacy budgets on influential points as shown in Figure~\ref{fig:membership-inference}.

\subsubsection{Sensitivity to Thresholds} When the threshold is set to a low value, such as 0.05, slight deviations in survival probability due to differential privacy noise can cause a time point to be flagged as influential. The added Laplace noise varies randomly in each trial, thus causing fluctuations in the number of influential points. As the threshold increases (e.g., to 0.5 or 0.7), the impact of smaller fluctuations diminishes, leading to fewer time points surpassing the threshold and more stable influential point counts across trials.

\subsubsection{Impact of Privacy Budget} As the privacy budget $\epsilon$ increases, the noise added for differential privacy decreases, leading to closer alignment between the survival estimates with and without the target data point. This reduction in noise makes it less likely for survival probability differences to exceed higher thresholds, resulting in fewer influential points. For example, at $\epsilon = 10$, the DP-KM curves closely resemble the non-private survival curve, and we see fewer influential points even at lower thresholds.

\subsection{Interpretation of Results}
The results indicate that at lower thresholds (e.g., 0.05), the membership inference attack becomes more sensitive to random noise fluctuations introduced by differential privacy, leading to greater variability in influential point counts across trials. This variability arises because each trial introduces new noise, causing slight differences in survival probabilities. At lower thresholds, these small differences are often enough to surpass the threshold, leading to inconsistencies in the number of influential points detected.

In practical terms, a high variability in influential points suggests that detecting the membership status of an individual at very low thresholds is unreliable due to the random noise inherent in differential privacy. Consequently, low thresholds might lead to false positives in membership inference, where time points are flagged as influential due to random noise rather than actual influence from the target data point.

In summary:
\begin{itemize}
    \item \textbf{Lower thresholds} increase the sensitivity to noise, resulting in more variation across trials and a higher likelihood of false positives.
    \item \textbf{Higher thresholds} reduce the effect of minor noise fluctuations, stabilizing the number of influential points and providing stronger protection against membership inference attacks.
\end{itemize}

\section{Summary of Results}

\subsection{Privacy-Utility Trade-off Analysis}
\textbf{Objective}: Assess how varying privacy budgets ($\epsilon$) affect the accuracy of survival probability estimates.

\textbf{Key Findings}: As $\epsilon$ increases, the RMSE between the differentially private (DP) and non-private Kaplan-Meier (KM) estimates decreases, indicating improved accuracy with higher $\epsilon$ values (less privacy). With lower $\epsilon$ values, survival probabilities experience higher noise, leading to increased RMSE and greater deviation from non-private estimates.

\textbf{Interpretation}: This result reveals a clear privacy-utility trade-off: low $\epsilon$ values introduce significant noise, causing greater distortion in survival estimates, while high $\epsilon$ values reduce noise impact, improving utility at the cost of privacy.

\subsection{Effect of Parameters on Differentially Private KM Curves}

\textbf{Objective}: Investigate how $\alpha$ (time-indexed noise scaling), $\tau_{\text{start}}$ and $\tau_{\text{end}}$ (dynamic clipping thresholds), and $w$ (smoothing window size) impact DP KM curve stability and accuracy.

\begin{itemize}
    \item \textbf{Noise Scaling Factor ($\alpha$)}: 
    Increasing $\alpha$ improves the stability of the survival curve by progressively reducing the amount of noise added at later time points, where survival probabilities are lower. In contrast, lower $\alpha$ values apply a uniform level of noise across all time points, which can result in higher variability, particularly at early time points when the number of surviving individuals is larger or in a dataset with a few individuals failing at early time points. 
  
    \item \textbf{Clipping Thresholds ($\tau_{\text{start}}$, $\tau_{\text{end}}$)}: Setting a high $\tau_{\text{start}}$ (near 1) with a moderately lower $\tau_{\text{end}}$ provides a realistic survival curve shape while constraining extreme values caused by noise. With lower clipping thresholds, noise effects can distort survival probabilities at early time points, whereas too high thresholds make survival estimates less sensitive to individual variations, diminishing privacy.
    \item \textbf{Smoothing Window Size ($w$)}: Increasing $w$ improves smoothness of the DP KM curves, reducing noise fluctuation but potentially averaging out real trends in the data. A small $w$ retains more granular survival trends but may show higher variability due to noise.
\end{itemize}

\textbf{Interpretation}: Fine-tuning these parameters allows the balance of privacy and utility by controlling noise intensity, clipping levels, and curve smoothness. When a target error tolerance threshold for RMSE is defined, the optimal settings for parameters (e.g., $\alpha$, $\tau_{\text{start}}$, $\tau_{\text{end}}$, and $w$) can be determined using various optimization algorithms to meet the given target with minimal deviations from non-private estimates. By fine-tuning $\alpha$ for time-indexed noise decay, adjusting $\tau_{\text{start}}$ and $\tau_{\text{end}}$ for dynamic clipping, and selecting an appropriate smoothing window size $w$, these parameters can be optimized to achieve a stable differentially private Kaplan-Meier curve that aligns closely with the desired error tolerance.

\subsection{Comparison of Non-Private and Differentially Private Kaplan-Meier Curves}

\textbf{Objective}: Evaluate alignment between DP and non-private KM curves across different $\epsilon$ values.

\textbf{Key Findings}: 
\begin{itemize}
    \item At low $\epsilon$ values (e.g., $\epsilon = 0.1$), DP KM curves exhibit frequent zero values at later time points due to high noise, deviating significantly from non-private estimates.
    \item As $\epsilon$ increases, DP KM curves align more closely with non-private estimates, especially at $\epsilon = 10$, which shows minimal deviation across all time points.
    \item Intermediate $\epsilon$ values (e.g., $\epsilon = 4$, $6$, and $8$) maintain stable survival estimates without sharp fluctuations, though remaining slightly below non-private probabilities.
\end{itemize}

\textbf{Interpretation}: Higher $\epsilon$ values reduce noise impact, enabling DP KM curves to reflect true survival trends. This finding reinforces the privacy-utility trade-off: higher privacy budgets provide closer alignment with non-private curves, yielding better utility at a reduced privacy level.

\subsection{Sensitivity to Influential Points and Membership Inference Attack Success}

\textbf{Objective}: Assess model vulnerability to membership inference attacks by identifying influential points based on thresholds.

\textbf{Key Findings}:
\begin{itemize}
    \item Lower $\epsilon$ values introduce high noise, leading to frequent influential points where survival probabilities deviate at key time points. Small thresholds (e.g., 0.05) yield more influential points across trials due to greater sensitivity to noise effects.
    \item Higher $\epsilon$ values (e.g., $\epsilon \geq 6$) reduce the number of influential points, especially at higher thresholds, lowering susceptibility to membership inference attacks.
\end{itemize}

\textbf{Interpretation}: The presence of influential points highlights potential privacy risks, as these points reveal membership sensitivity. Lower $\epsilon$ values introduce noise that makes survival estimates more susceptible to attack, while higher $\epsilon$ values reduce this risk by aligning more closely with non-private curves, making membership harder to infer.

\subsection{General Observations Across Experiments}

\begin{itemize}
    \item \textbf{Low privacy budgets} ($\epsilon = 0.1$ to $2$) lead to high noise, resulting in high RMSE, fluctuating survival probabilities, and greater vulnerability to membership inference attacks.
    \item \textbf{Moderate privacy budgets} ($\epsilon = 4$ to $6$) balance utility and privacy, providing stable survival estimates with few influential points, while maintaining alignment with non-private estimates.
    \item \textbf{High privacy budgets} ($\epsilon = 8$ and above) offer the best alignment with non-private KM curves, minimal noise impact, and few influential points, optimizing utility at reduced privacy protection.
\end{itemize}

The combined privacy-utility analysis (Figure~\ref{fig:sensitivity-analysis}) and DP KM curve comparisons (Figure~\ref{fig:km-comparison}) confirm that carefully balancing $\epsilon$ and other model parameters reduces RMSE, aligns DP survival estimates with non-private data, and mitigates the privacy risks associated with influential points.

\section{Related Work}

Differential privacy has been widely studied to protect sensitive information across various machine learning and statistical models, particularly against membership inference attacks, where adversaries aim to determine if a specific record was included in the training dataset.

Gondara \cite{gondara2020differentially} was among the first to explore differential privacy in survival analysis by adding uniform noise to survival probability estimates, which provided privacy guarantees. However, Gondara's approach did not adapt noise levels to changes in the at-risk population size, resulting in potential over- or under-estimations as survival probabilities evolve over time. In contrast, our approach leverages time-indexed noise scaling, dynamic clipping, and smoothing, ensuring the differential privacy noise adapts to changes in the cumulative structure of the Kaplan-Meier estimator, producing more accurate and stable survival curves.

The membership inference attack framework by Shokri et al. \cite{shokri2017membership} introduced shadow models to estimate membership inference vulnerabilities in machine learning models. Although effective in detecting privacy risks, Shokri’s method focuses solely on identifying privacy risks post hoc and does not modify the models to mitigate them proactively. Our work diverges by embedding differential privacy mechanisms directly into the Kaplan-Meier estimation, reducing membership inference risks through adaptive noise and clipping adjustments instead of relying on detection alone.

In \cite{fan2023mitigating}, Fan and Bonomi apply differential privacy to deep learning-based survival models like DeepSurv and DeepHit, implementing noise addition and gradient clipping during training. While this method effectively maintains privacy in neural survival models, it is complex and less interpretable when applied to traditional survival analysis methods. Our approach differs by focusing on the Kaplan-Meier estimator, a widely-used, interpretable method in medical and clinical contexts, thereby providing privacy-preserving survival analysis with clearer interpretability.

Bonomi et al. \cite{bonomi2020protecting} also address privacy concerns in survival analysis by incorporating noise addition techniques. However, their approach is primarily suited for protecting individual patient data in aggregate models and does not incorporate adaptive mechanisms based on the dynamic population at risk. In contrast, our method's adaptive noise scaling and clipping thresholds dynamically adjust to the at-risk population size over time, enhancing the stability and accuracy of the privacy-preserving Kaplan-Meier estimator.

Lastly, Winograd-Cort et al. \cite{winograd2017framework} present the Adaptive Fuzz framework, which introduces adaptive differential privacy composition for flexible application across different privacy queries. While this framework is instrumental for adaptive applications and allows for dynamic privacy management, it does not specifically target survival analysis or Kaplan-Meier estimation. Our approach integrates similar adaptive principles but is specifically tailored for survival analysis, using Kaplan-Meier curves with mechanisms tuned for time-indexed survival probabilities.

In summary, while existing works cover privacy in survival analysis through uniform noise addition, complex neural models, and adaptive frameworks, our approach offers a refined and focused application for Kaplan-Meier estimations. By introducing adaptive time-indexed noise scaling, dynamic clipping, and smoothing, our method maintains the cumulative nature of Kaplan-Meier survival estimates and addresses the specific privacy challenges in survival analysis.

\section{Conclusion}

In this work, we addressed key limitations in privacy-preserving survival analysis by proposing a novel approach that integrates time-indexed noise scaling, adaptive clipping, and smoothing mechanisms. These innovations effectively preserve the cumulative structure of Kaplan-Meier estimates while significantly reducing root mean squared error (RMSE). Our method demonstrated notable improvements in both privacy and utility on the NCCTG lung cancer dataset, particularly under higher privacy budgets (e.g., $\epsilon \geq 6$), where it effectively minimized the influence of outliers and enhanced resilience against membership inference attacks.

Despite these advancements, our approach has certain limitations that warrant further exploration. The dependence of the total privacy budget $\hat{\varepsilon} = O(\epsilon \cdot \alpha n^2)$ on the total number of checkpoints $n$ is an area for improvement. Refining the scheme and analysis could reduce this dependency and optimize privacy usage. Additionally, incorporating the smoothing mechanism directly into the Laplace mechanism offers the potential to enhance the privacy-utility tradeoff. Finally, testing our system against more adversarial membership inference attacks would provide a clearer understanding of its robustness to information leakage.

Nonetheless, the proposed method already demonstrates strong empirical performance on a realistic dataset (NCCTG lung cancer dataset), bridging the gap between practical utility and theoretical guarantees. We believe that future iterations can further close this gap, advancing the field of privacy-preserving survival analysis.

\bibliography{sn-bibliography}

\end{document}